\documentclass[12pt,reqno,a4paper]{amsart}
\usepackage{amsmath,amssymb,dsfont,graphicx,cite,latexsym,
epsf,cancel,epic,eepic}
\usepackage[colorlinks=false]{hyperref}
\usepackage{mathpazo}
\usepackage{tikz-cd}

\newtheorem{theorem}{Theorem}
\newtheorem{proposition}[theorem]{Proposition}
\newtheorem{lemma}[theorem]{Lemma}

\newcommand{\bbR}{{\mathbb R}}

\def\t{\widetilde}

\title[How one can repair non-integrable Kahan discretizations. II]{How one can repair non-integrable Kahan discretizations. \\ II. A planar system with invariant curves of degree 6}
\author{Misha Schmalian}
\address{Trinity College, Trinity Street, CB21TQ, Cambridge, UK}
\email{mjs297@cam.ac.uk}

\author{Yuri B. Suris}
\address{Institut f\"ur Mathematik, MA 7-1,
Technische Universit\"at Berlin, Str. des 17. Juni 136,
10623 Berlin, Germany}
\email[Corresponding author]{suris@math.tu-berlin.de}

\author{Yuriy Tumarkin}
\address{Trinity College, Trinity Street, CB21TQ, Cambridge, UK}
\email{yt354@cam.ac.uk}

\begin{document}

\maketitle

\begin{abstract}
We find a novel one-parameter family of integrable quadratic Cremona maps of the plane preserving a pencil of curves of degree 6 and of genus 1. They turn out to serve as Kahan-type discretizations of a novel family of quadratic vector fields possessing a polynomial integral of degree 6 whose level curves are of genus 1, as well. These vector fields are non-homogeneous generalizations of reduced Nahm systems for magnetic monopoles with icosahedral symmetry, introduced by Hitchin, Manton and Murray. The straightforward Kahan discretization of these novel non-homogeneous systems is non-integrable. However, this drawback is repaired by introducing adjustments of order $O(\epsilon^2)$ in the coefficients of the discretization, where $\epsilon$ is the stepsize.
\end{abstract}

%%%%%%%%%%%%%%%%%%%%%%%%%%%%%%%
%%%%%%%%%%%%%%%%%%%%%%%%%%%%%%%
\section{Introduction}
%%%%%%%%%%%%%%%%%%%%%%%%%%%%%%%
%%%%%%%%%%%%%%%%%%%%%%%%%%%%%%%

The problem of integrable discretization \cite{S} consists of finding, for a given integrable system, a discretization which remains integrable. All conventional discretization methods for ODEs, like Runge-Kutta methods etc., fail to preserve integrability. However, there exists an ``unconventional'' numerical method applicable to any system of ODEs on $\bbR^n$ with a quadratic vector field, known as Kahan discretization, which possesses remarkable properties in this respect. Consider a quadratic ODE
\begin{equation}\label{eq: diff eq gen}
\dot{x}=Q(x)+Bx+c,
\end{equation}
where $Q:\bbR^n\to\bbR^n$ is a vector of quadratic forms, $B$ is an $n\times n$ matrix, and $c\in\bbR^n$. Kahan discretization, introduced in \cite{K}, consists in replacing the time derivative on the left-hand side by the first difference of the numerical approximation $x:\epsilon\mathbb Z\to \mathbb R^n$, while the quadratic expressions on the right-hand side are replaced by symmetric bilinear expressions in terms of $x=x(t)$ and $\t x=x(t+\epsilon)$:
\begin{equation}\label{eq: Kahan gen}
\frac{\widetilde{x}-x}{\epsilon}=Q(x,\widetilde{x})+\frac{1}{2}B(x+\widetilde{x})+c,
\end{equation}
where
\[
Q(x,\widetilde{x})=\frac{1}{2}\big(Q(x+\widetilde{x})-Q(x)-Q(\widetilde{x})\big)
\]
is the symmetric bilinear form corresponding to the quadratic form $Q$. Equation (\ref{eq: Kahan gen}) is {\em linear} with respect to $\widetilde x$ and therefore defines a {\em rational} map $\widetilde{x}=f(x,\epsilon)$. Due to the symmetry of equation (\ref{eq: Kahan gen}) with respect to interchanging $x\leftrightarrow\widetilde{x}$ accompanied by sign inversion $\epsilon\mapsto-\epsilon$, the map $f$ is {\em reversible}:
\begin{equation}\label{eq: reversible}
f^{-1}(x,\epsilon)=f(x,-\epsilon).
\end{equation}
Thus, the map $f$ is {\em birational}. For some reasons which remain not completely clarified up to now, Kahan's method tends to preserve integrability much more often than any other known general purpose discretization scheme.

This was first observed by Hirota and Kimura \cite{HK, KH}, who (being unaware of the work by Kahan) applied this scheme to the Euler top and to the Lagrange top, and observed that the resulting maps are integrable. Since then, integrability properties of Kahan's method when applied to integrable systems (also called ``Hirota-Kimura method'' in this context) were extensively studied, see \cite{PS, PPS1, PPS2, PPS3, PS2, PS3, PS4, PZ, PSS, Z, PSZ} and \cite{CMOQ1, CMOQ2, CMOQ4, CMOQ5, KCMMOQ, KMQ}. Integrability is preserved in an amazing number of cases, but not always.

Simple counterexamples are available already in dimension $n=2$, and can be found among non-homogeneous extensions of the so called reduced Nahm equations introduced in \cite{HMM}. These are the systems of the form
\begin{equation} \label{nahm intro}
\begin{pmatrix} \dot x\\ \dot y \end{pmatrix}
 =\frac{1}{\rho(x,y)}\begin{pmatrix} \partial H/\partial y  \\ -\partial H/\partial x \end{pmatrix},
\end{equation}
where
\begin{equation*}
H(x,y)=\ell_1^{\gamma_1}(x,y)\ell_2^{\gamma_2}(x,y)\ell_3^{\gamma_3}(x,y), \quad \rho(x,y)=\ell_1^{\gamma_1-1}(x,y)\ell_2^{\gamma_2-1}(x,y)\ell_3^{\gamma_3-1}(x,y),
\end{equation*}
with $\gamma_1, \gamma_2, \gamma_3\in\bbR\setminus\{0\}$, and $\ell_i(x,y)=a_ix+b_iy$
are linear forms. Integrability takes place for $(\gamma_1,\gamma_2,\gamma_3)=(1,1,1)$, $(1,1,2)$, and $(1,2,3)$. In all three cases, all integral curves of the system \eqref{nahm intro} are of genus 1. In \cite{PPS2, PZ, CMOQ4} integrabilty was established for the Kahan discretization of all three cases of the reduced Nahm equations.

If $(\gamma_1,\gamma_2,\gamma_3)=(1,1,1)$, one is dealing with a homogeneous cubic Hamiltonian. As discovered in \cite{CMOQ1}, Kahan's discretization remains integrable for arbitrary (i.e., also for non-homogeneous) cubic Hamiltonians.

If $(\gamma_1,\gamma_2,\gamma_3)=(1,1,2)$, one can find non-homogeneous perturbations of the quartic polynomial $H(x,y)$ so that the resulting differential equations \eqref{nahm intro} still have the above mentioned property: all integral curves are of genus 1. A Kahan discretization of the perturbed  (non-homogeneous) system is non-integrable. However, it was shown in \cite{PSZ} that one can adjust the coefficients of the discretization (making them dependent on $\epsilon$ in a non-trivial way) to obtain an integrable Kahan-type discretization.

The present paper is devoted to a similar result for systems of the class $(\gamma_1,\gamma_2,\gamma_3)=(1,2,3)$. The homogeneous system can be taken as
\begin{equation} \label{123 intro}
\left\{\begin{array}{l} \dot x=-2x^2+2xy, \\ \dot{y}=-y^2+2xy. \end{array}\right.
\end{equation}
It possesses an integral of motion of degree 6:
\begin{equation}\label{123 H intro}
H(x,y)=x^2y^3\Big(-\frac{2}{3}x+\frac{1}{2}y\Big),
\end{equation}
whose level sets are curves of genus 1.
The Kahan discretization of this system reads:
\begin{equation} \label{123 map intro}
\renewcommand{\arraystretch}{1.2}
\left\{\begin{array}{l}
 (\t x - x)/  \epsilon=-2\t x x+(\t x y+x\t y), \\
 (\t y - y)/ \epsilon= -\t y y +(\t  xy + x\t y).
\end{array} \right .
\end{equation}
It is integrable, with an integral of motion
\begin{equation}\label{123 integral intro}
H_1(x,y)=\dfrac{H(x,y)}{\big(1-\epsilon^2x^2\big)\big(1-\epsilon^2(x-y)^2\big)\big(1-\epsilon^2(x^2+y^2)\big)}.
\end{equation}

Consider the following non-homogeneous perturbation of system \eqref{123 intro}:
\begin{equation}\label{123 pert intro}
\left\{\begin{array}{l} \dot x= -2x^2+2xy+c, \\ \dot y=-y^2+2xy. \end{array}
\right.
\end{equation}
It has the following integral of motion:
\begin{equation}\label{H pert intro}
H(x,y)  = (xy+c)^2\left(-\frac{2}{3}xy+\frac{1}{2}y^2+\frac{1}{3}c\right),
\end{equation}
with the same property as above (all level sets are curves of genus 1). The Kahan discretization of this system,
\begin{equation}\label{123 pert Kahan}
\left\{\begin{aligned}
(\t x - x)/\epsilon &= -2 x \t x + \left(\t x y + x \t y\right) + c,\\
(\t y - y)/\epsilon &= - y \t y + (\t x y + x \t y),
\end{aligned}
\right.
\end{equation}
generates a non-integrable map. However, the coefficients of this discretization can be adjusted via $O(\epsilon^2)$ terms, to produce an integrable map:
\begin{equation} \label{Kahan adj intro}
\left\{\begin{aligned}
(\t x - x)/\epsilon &= -(2-\epsilon^2c)x \t x  + (1+\epsilon^2c)(\t x y + x \t y)+ c- \epsilon^2c(2+\epsilon^2c) y \t y,\\
(\t y - y)/\epsilon &= - (1+\epsilon^2c)y \t y + (\t x y + x \t y).
\end{aligned}
\right.
\end{equation}
This map, like the unperturbed one \eqref{123 map intro}, has an integral of motion whose level sets are curves of degree 6 and of genus 1 (the irreducible ones).

The presentation is organized as follows. In Section \ref{sect 123}, we consider in detail system \eqref{123 intro} and its Kahan discretization \eqref{123 map intro}, paying special attention to the singularity confinement property of the latter map. In Section \ref{sect QRT root}, we perform, following \cite{PSWZ}, a reduction of the pencil of invariant curves of degree 6 of the Kahan discretization to a pencil of biquadratic curves. This way, the map is shown to be birationally equaivalent to a special QRT root (cf. \cite{QRT, Dui}). In Section \ref{sect gen QRT root}, we show that the relevant geometric and dynamical properties of this QRT root can be found in a one-parameter family of such maps, and then find a corresponding one-parameter family of birationally equivalent Kahan-type maps preserving a pencil of curves of degree 6 and of genus 1. Finally, in Section \ref{sect cont limit}, a continuous limit is performed in those Kahan-type maps, leading to a novel integrable system \eqref{123 pert intro}, with a pencil of invariant curves with the same property (level sets of the non-homogeneous sextic polynomial \eqref{H pert intro} are of genus 1).

%%%%%%%%%%%%%%%%%%%%%%%%%%%%%%%%%%%%%%%
\section{A homogeneous (1,2,3) system and its Kahan discretization}
\label{sect 123}

We start with a reduced Nahm system \eqref{123 intro} with $(\gamma_1,\gamma_2,\gamma_3)=(1,2,3)$, obtained by the following choice of the corresponding linear forms:
$$
\ell_1(x,y)=-\frac{2}{3}x+\frac{1}{2}y, \quad \ell_2(x,y)=x, \quad \ell_3(x,y)=y,
$$
so that $H(x,y)$ is as given in \eqref{123 H intro}, and $\rho(x,y)=xy^2$.

Its Kahan discretization is given in \eqref{123 map intro}. Due to homogeneity, we can restrict ourselves to the case $\epsilon=1$,
\begin{equation} \label{123 map}
\renewcommand{\arraystretch}{1.2}
\left\{\begin{array}{l}
 \t x - x=-2\t x x+(\t x y+x\t y), \\
 \t y - y= -\t y y +(\t  xy + x\t y).
\end{array} \right .
\end{equation}
The general case is obtained from this by the re-scaling $(x,y)\mapsto(\epsilon x,\epsilon y)$. A simple computation gives an explicit formula for the map $f$:
$$
\begin{pmatrix} \t x \\ \t y \end{pmatrix}=\begin{pmatrix}  1+2x-y & -x \\ -y & 1-x+y\end{pmatrix}^{-1}\begin{pmatrix} x \\ y \end{pmatrix}
=\frac{1}{\Delta}\begin{pmatrix}  1-x+y & x \\ y & 1+2x-y\end{pmatrix}\begin{pmatrix} x \\ y \end{pmatrix},
$$
or
\begin{equation}\label{123 map explicit}
\t x= \frac{x(1-x+2y)}{\Delta}, \quad \t y = \frac{y(1+3x-y)}{\Delta}, \quad \Delta=1+x-2x^2+2xy-y^2.
\end{equation}
In homogeneous coordinates,
\begin{equation}\label{123 map homog}
f(x,y,z)= \Big[x(z-x+2y):  y(z+3x-y):  z^2+zx-2x^2+2xy-y^2\Big].
\end{equation}
In the following proposition, we collect the relevant information about this map, as found in \cite{PPS2, PZ, CMOQ4,Z}.
\begin{proposition}
The map $f$ given in \eqref{123 map explicit} admits an integral of motion
\begin{equation}\label{123_integral}
H_1(x,y)=\dfrac{H(x,y)}{\big(1-x^2\big)\big(1-(x-y)^2\big)\big(1-(x^2+y^2)\big)}\ ,
\end{equation}
with $H(x,y)=x^2y^3(-\frac{2}{3}x+\frac{1}{2}y)$. The pencil of the level curves $H_1(x,y)=\lambda$, i.e.,
\begin{equation}\label{123 pencil}
H(x,y)-\lambda\big(1-x^2\big)\big(1-(x-y)^2\big)\big(1-(x^2+y^2)\big)=0,
\end{equation}
of $\deg=6$ possesses eleven (distinct) base points given by:
\begin{itemize}
\item six finite base points of multiplicity $1$ on the line $\ell_1=0$:
\begin{equation}\label{123_p1-p6}
p_1=-p_6=\Big(\frac{3}{5},\frac{4}{5}\Big),\quad
p_2=-p_5=\Big(1,\frac{4}{3}\Big), \quad
p_3=-p_4=(3,4),
\end{equation}

\item three base points of multiplicity $2$ on the line $\ell_2=0$, two finite and one at infinity:
\begin{equation}\label{123_p7-p9}
p_7=-p_9=(0,-1), \quad p_8=[0:1:0],
\end{equation}

\item and two finite base points of multiplicity $3$ on the line $\ell_3=0$:
\begin{equation}\label{123_p10-p11}
p_{10}=-p_{11}=(-1,0).
\end{equation}
\end{itemize}
See Fig. \ref{Fig pencil 123} for an illustration.
One has: $\mathcal I(f)=\{p_6,p_9,p_{11}\}$ and $\mathcal I(f^{-1})=\{p_1,p_7,p_{10}\}$. All base points participate in three confined singular orbits of the map $f$:
\begin{eqnarray}
\begin{aligned}\label{123_orbits}
&(p_9p_{11}) \longrightarrow p_1 \longrightarrow p_2 \longrightarrow p_3 \longrightarrow p_4 \longrightarrow p_5 \longrightarrow p_6 \longrightarrow (p_7p_{10}),\\
&(p_6p_{11}) \longrightarrow p_7 \longrightarrow p_8 \longrightarrow p_9 \longrightarrow (p_1p_{10}),\\
&(p_6p_9) \longrightarrow p_{10} \longrightarrow p_{11} \longrightarrow (p_1p_7).
\end{aligned}
\end{eqnarray}
\end{proposition}

%%%%%%%%%%%%%%%%%%%%%%%%%%%%%%%%%%%%%%%%%%%%%%%%%%%%%%%%%%%%%%%%%%%%
\begin{figure}
\begin{center}
\includegraphics[width=0.8\textwidth]{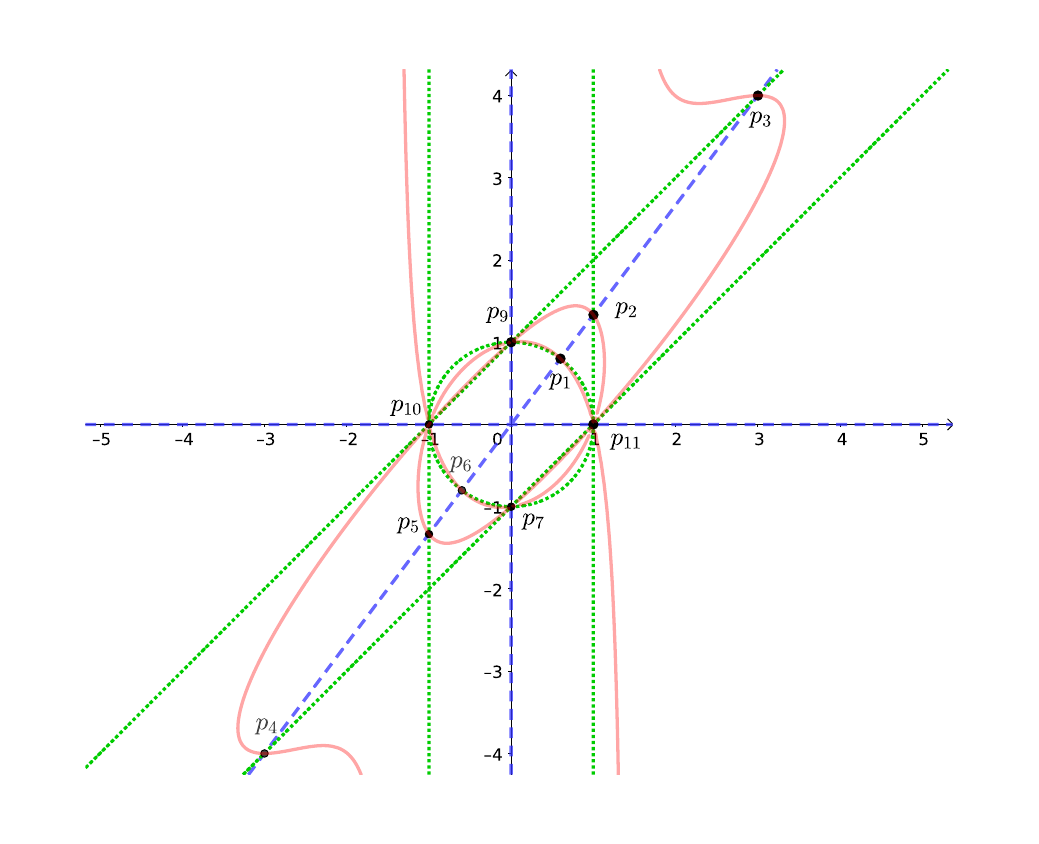}
\caption{Some invariant curves of the sextic pencil \eqref{123 pencil}. Blue: three lines (one of them double and one triple) -- the vanishing set of the numerator of $H_1(x,y)$, i.e., $\lambda=0$. Green: four lines and a conic -- the vanishing set of the denominator of $H_1(x,y)$, i.e., $\lambda=\infty$. Red: a generic curve.}
\label{Fig pencil 123}
\end{center}
\end{figure}
%%%%%%%%%%%%%%%%%%%%%%%%%%%%%%%%%%%%%%%%%%%%%%%%%%%%%%%%%%%%%%%%%%%%

We refer the reader to \cite{BV, D1,D2,DF} for general information about birational (Cremona) maps of $\mathbb C\mathbb P^2$, including the notion of confined singular orbits (related to degree-lowering curves and dynamical degree, or algebraic entropy).

%%%%%%%%%%%%%%%%%%%%%%%%%%%%%%%%
\section{Reduction of the map $f$ to a special QRT root}
\label{sect QRT root}

We use notation
\begin{equation}\label{pencil sextic}
{\mathcal E}_6={\mathcal P}(6;p_1,\ldots,p_6,p_7^2,p_8^2,p_9^2,p_{10}^3, p_{11}^3)
\end{equation}
for the pencil of curves of degree 6 with simple base boints $p_1,\ldots,p_6$, double base points $p_7,p_8,p_9$, and triple base points $p_{10},p_{11}$.
One can simplify such a pencil by applying a quadratic Cremona map $\phi$ with the fundamental points $p_9,p_{10},p_{11}$ (both the triple base points and one of the double base points), cf. \cite{PSWZ}.
\begin{proposition} \label{th sextic to quartic}
Consider a quadratic Cremona map $\phi$ blowing down the lines  $(p_{10}p_{11})$, $(p_9p_{11})$, $(p_9p_{10})$ to points denoted by $q_9,q_{10},q_{11}$, respectively, and blowing up the points $p_9,p_{10},p_{11}$ to the lines $(q_{10}q_{11})$, $(q_9q_{11})$, $(q_9q_{10})$. All other  base points $p_i$, $i=1,\ldots,8$ are regular points of $\phi$ and their images are denoted by $q_i=\phi(p_i)$. The change of variables $\phi$ maps pencil \eqref{pencil sextic} of sextic curves to the pencil
\begin{equation}\label{pencil quartic}
{\mathcal E}_4={\mathcal P}(4;q_1,\ldots,q_6,q_{10},q_{11},q_7^2,q_8^2)
\end{equation}
of quartic curves with eight simple base points and two double base points. The point $q_9$ is not a base point of the latter pencil.
\end{proposition}
\begin{proof}
 The total image of a curve $C\in{\mathcal E}_6$ is a curve of degree 12. Since $C$ passes through $p_9,p_{10},p_{11}$ with multiplicities 2,3,3, its total image contains the lines $(q_{10}q_{11})$, $(q_9q_{11})$, $(q_9q_{10})$ with the same multiplicities. Dividing by the linear defining polynomials of all these lines, we see that the proper image of $C$ is a curve of degree $12-8=4$. This curve passes through all points $q_i$, $i=1,\ldots,8$ (for $i=7,8$ with multiplicity 2). The curve $C$ of degree 6 has no other intersections with the line $(p_{10}p_{11})$ different from two triple points $p_{10}$ and $p_{11}$, therefore its proper image does not pass through $q_9$. On the other hand, the curve $C$ of degree 6 has one additional intersection point with each of the lines $(p_9p_{10})$ and $(p_9p_{11})$, different from the double point $p_9$ and the triple point $p_{10}$, respectively $p_{11}$. Therefore, its proper image passes through $q_{11}$, resp. $q_{10}$, with multiplicity 1.
\end{proof}

For the proof of the following Proposition, we will repeatedly use the following lemma.
\begin{lemma}\label{lemma Cremona}
Let $F$ be a quadratic Cremona map with $\mathcal I(F)=\{a,b,c\}$, and let $F$ blow down the lines $(ab)$, $(bc)$, $(ca)$ to the points $C$, $A$, $B$, respectively. Then the image of a generic line under $F$ is a conic through $A,B,C$. The (proper) image of a line through one of the indeterminacy points, say of the line $(ad)$, is  the line $(AD)$, where $D=F(d)$.
\end{lemma}
\begin{proof} The total image of the line $(ad)$ is a conic, but since $a$ is blown up to a line, the proper image is a line. This line has to pass through $D=F(d)$ and through $A$ (since the line $(ad)$ intersects the line $(bc)$ which is blown down to $A$).
\end{proof}

\begin{proposition}
The map $g=\phi\circ f\circ \phi^{-1}$ has three confined singular orbits:
\begin{eqnarray}
\begin{aligned}\label{123_QRT_orbits}
&(q_6q_8) \longrightarrow q_{10}\longrightarrow q_1 \longrightarrow q_2 \longrightarrow q_3 \longrightarrow q_4 \longrightarrow q_5 \longrightarrow q_6 \longrightarrow (q_7q_{10}),\\
&(q_6q_{11}) \longrightarrow q_7 \longrightarrow q_8 \longrightarrow (q_{10}q_{11}),\\
&(q_8q_{11}) \longrightarrow q_{11} \longrightarrow  (q_7q_{11}).
\end{aligned}
\end{eqnarray}
The point $q_9$ is its fixed point, and lies on the line $(q_7q_8)$. Moreover, the points $q_3$ and $q_{11}$ are infinitely near.
\end{proposition}
\begin{proof}
We have:
$$
q_i\xrightarrow{\phi^{-1}} p_i\overset{f}{\longrightarrow} p_{i+1}\overset{\phi}{\longrightarrow} q_{i+1}, \quad i=1,\ldots,5.
$$
Further,
$$
q_6\xrightarrow{\phi^{-1}} p_6\overset{f}{\longrightarrow} (p_7p_{10})\overset{\phi}{\longrightarrow} (q_7q_{10})
$$
(applying Lemma \ref{lemma Cremona} for $\phi$);
$$
q_7\xrightarrow{\phi^{-1}} p_7\overset{f}{\longrightarrow} p_8\overset{\phi}{\longrightarrow} q_8;
$$
$$
q_8\xrightarrow{\phi^{-1}} p_8\overset{f}{\longrightarrow} p_9\overset{\phi}{\longrightarrow} (q_{10}q_{11});
$$
$$
q_9\xrightarrow{\phi^{-1}} (p_{10}p_{11})\overset{f}{\longrightarrow} (p_{10}p_{11})\overset{\phi}{\longrightarrow} q_9
$$
(applying Lemma \ref{lemma Cremona} for $f$);
$$
q_{10}\xrightarrow{\phi^{-1}} (p_9p_{11})\overset{f}{\longrightarrow} p_1\overset{\phi}{\longrightarrow} q_1;
$$
$$
q_{11}\xrightarrow{\phi^{-1}} (p_9p_{10})\overset{f}{\longrightarrow} (p_7p_{11})\overset{\phi}{\longrightarrow} (q_7q_{11})
$$
(applying Lemma \ref{lemma Cremona} for $f$, then for $\phi$).

Next, we consider lines which are blown down by $g$:
$$
(q_6q_8) \xrightarrow{\phi^{-1}} C(p_6,p_8,p_9,p_{10},p_{11}) \overset{f}{\longrightarrow} (p_9p_{11})\overset{\phi}{\longrightarrow} q_{10}
$$
(indeed, the total $f$-image of the conic is a curve of degree 4; however, three lines split off, being the blow-ups of $p_6$, $p_9$, $p_{11}$; thus, the proper image is the line through $f(p_8)=p_9$ and $f(p_{10})=p_{11}$);
$$
(q_6q_{11}) \xrightarrow{\phi^{-1}} (p_6p_{11}) \overset{f}{\longrightarrow} p_7\overset{\phi}{\longrightarrow} q_7
$$
(applying Lemma \ref{lemma Cremona} for $\phi^{-1}$);
$$
(q_8q_{11}) \xrightarrow{\phi^{-1}} (p_8p_{11}) \overset{f}{\longrightarrow} (p_9p_{10})\overset{\phi}{\longrightarrow} q_{11}
$$
(applying Lemma \ref{lemma Cremona} for $\phi^{-1}$, then for $f$).

The fact that $q_3$ and $q_{11}$ are infinitely near follows from the fact that $p_3\in (p_9p_{10})$, the latter line being blown down to $q_{11}$ by $\phi$.

It remains to show that $q_9\in(q_7q_8)$. For this, observe that the total $\phi$-image of $(p_7p_8)$ is the conic $C(q_7,q_8,q_9,q_{10},q_{11})$. However, since $p_9\in (p_7p_8)$, the blow-up of $p_9$ splits off this conic. This is the line $(q_{10}q_{11})$, and it does not contain any of the points $q_7$, $q_8$, $q_9$. Thus, the proper $\phi$-image of $(p_7p_8)$ is a line containing the latter three points, which are therefore collinear.
\end{proof}

%To compute the map $\phi$, one first computes the net of homaloids -- the quadrics through $p_9$, $p_{10}$, $p_{11}$. This net is spanned by
%$$
%\phi_1(x,y,z)=x^2-z^2+yz, \quad \phi_2(x,y,z)=y^2-yz, \quad \phi_3(x,y,z)=xy.
%$$
%The map $\phi$ has the form $\phi(x,y,z)=[u:v:w]$, where
%$$
%u=a_1\phi_1+b_1\phi_2+c_1\phi_3, \quad v= a_2\phi_1+b_2\phi_2+c_2\phi_3, \quad w=a_3\phi_1+b_3\phi_2+c_3\phi_3.
%$$
%We choose the constants $a_k,b_k,c_k$, $k=1,2,3$, by requiring that pencil  \eqref{pencil quartic} consists of  symmetric biquadratics,
%i.e., its double points are
%$$
%q_7=\phi(p_7)=[0:1:0], \quad q_8=\phi(p_8)=[1:0:0],
%$$
%while the points
%$$
%q_9=\phi((p_{10}p_{11}))=\phi(\{y=0\}), \quad q_{11} =\phi((p_{9}p_{10}))=\phi(\{y-x=1\})
%$$
%lie on the symmetry axis $u=v$. We can prescribe $q_{11}$ arbitrarily, and we set $q_{11}=(-1,-1)=[-1:-1:1]$. As for $q_9$, we have to take it as the intersection of the line $(q_7q_8)=\{w=0\}$ with the symmetry axis $u=v$, so that $q_9=[1:1:0]$. We arrive at the conditions
%\begin{align*}
%q_7=\phi([0:-1:1])=[0:1:0]: & \qquad a_1-b_1=0, \quad a_3-b_3=0;\\
%q_8=\phi([0:1:0])=[1:0:0]: & \qquad b_2=0, \quad b_3=0;\\
%q_9=\phi(\{y=0\})=[1:1:0]: & \qquad a_1=a_2, \quad a_3=0;\\
%q_{11} =\phi(\{y-x=z\})=[-1:-1:1]:&  \qquad a_1+b_1+c_1=a_2+b_2+c_2=-(a_3+b_3+c_3);
%\end{align*}
%We find: $a_1=b_1=a_2=a$ (say), $a_3=b_3=b_2=0$, $a+c_1=c_2$, and $a+c_2=-c_3$. This still leaves us with one free parameter, which we choose as $c_2=0$, then $c_1=-a$ and $c_3=-a$, so that finally

For an actual computation of the map $\phi$, we  can assume, without loss of generality, that pencil  \eqref{pencil quartic} consists of  symmetric biquadratics, i.e., its double points are
$$
q_7=\phi(p_7)=[0:1:0], \quad q_8=\phi(p_8)=[1:0:0],
$$
while the points
$$
q_9=\phi((p_{10}p_{11}))=\phi(\{y=0\}), \quad q_{11} =\phi((p_{9}p_{10}))=\phi(\{y-x=1\})
$$
lie on the symmetry axis $u=v$. This still leaves us with one free parameter. It can be chosen so that
\begin{equation}\label{phi}
\phi(x,y,z) = \big[  z^2-x^2-y^2+x y :  z^2-x^2-y z : x y \big].
\end{equation}
A direct computation with this formula gives:
$$
q_1=\big(1,-\tfrac{1}{3}\big), \quad q_2=\big(-\tfrac{1}{3},-1\big), \quad q_4=\big(-1,-\tfrac{1}{3}\big), \quad q_5=\big(-\tfrac{1}{3},1\big),
$$
$$
q_6=(1,3), \quad q_7=[0:1:0], \quad q_8=[1:0:0], \quad q_9=[1:1:0], \quad q_{10}=(3,1),
$$
$$
q_{11}=(-1,-1), \quad q_3>q_{11}\; (\text{slope}\; -1).
$$
The latter notation means that the point $q_3$ is infinitely near to $q_{11}$ and corresponds to the tangent line $\{v=-u-2\}$ there.

Now it remains to compute the map $g=\phi^{-1}\circ f\circ \phi$, i.e., the map $f$ in the coordinates $[u:v:w]$. A direct computation shows that, in the non-homogeneous coordinates, $g(u,v,1)=[\t u: \t v :1]$ with
\begin{equation}\label{QRT root}
 \t u=v, \quad \t v=\frac{uv-u-2}{2u-v+1},
\end{equation}
and admits an integral of motion
\begin{equation}\label{integral QRT root}
K(u,v)=\frac{3(u-v)^2-2(u+v)-4}{(u^2-1)(v^2-1)}.
\end{equation}
Thus, all base points lie on the four lines $\{u=\pm 1\}$, $\{v=\pm 1\}$, while the eight finite base points $q_1,\ldots,q_6,q_{10},q_{11}$ lie on the conic (parabola) $\{3(u-v)^2-2(u+v)-4=0\}$. See Fig. \ref{Fig biquadratics}.

%%%%%%%%%%%%%%%%%%%%%%%%%%%%%%%%%%%%%%%%%%%%%%%%%%%%%%%%%%%%%%%%%%%%
\begin{figure}
\begin{center}
\includegraphics[width=0.6\textwidth]{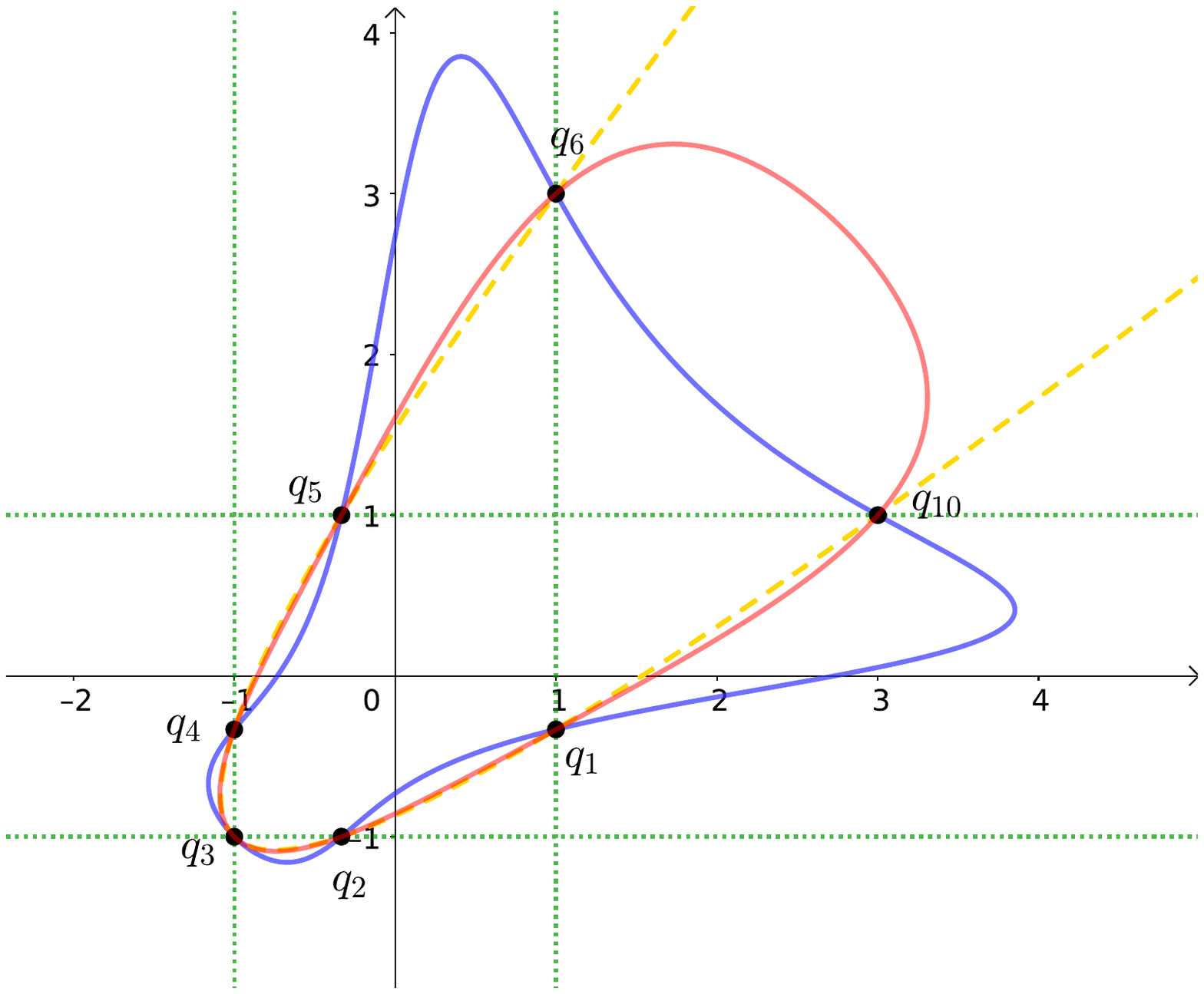}
\caption{Two biquadratics of the pencil $\{K(u,v)=\lambda\}$ with $K(u,v)$ from \eqref{integral QRT root}. All biquadratics pass through $q_{11}=(-1,-1)$ with the slope $-1$ (corresponding to $q_3$). Yellow: the conic through the eight base points $q_1,\ldots,q_6,q_{10},q_{11}$. Green: four lines $\{u=\pm 1\}$, $\{v=\pm 1\}$.}
\label{Fig biquadratics}
\end{center}
\end{figure}
%%%%%%%%%%%%%%%%%%%%%%%%%%%%%%%%%%%%%%%%%%%%%%%%%%%%%%%%%%%%%%%%%%%%

%%%%%%%%%%%%%%%%%%%%%%%%%%%%%%%%%%%%%%%%%%%%%%%%%%%
\section{Generalization of the QRT root}
\label{sect gen QRT root}

We try to generalize the map of the previous section. All objects found here will be one-parameter perturbations (with the parameter $c$) of the corresponding objects from the previous section. We will refrain from indicating this by an extra $c$ in the notation (to keep it as brief as possible). However, the reader should keep in mind that the unperturbed situation corresponds to $c=0$.

The idea is to stay in the class of symmetric QRT roots of $\deg=2$: in non-homogeneous coordinates, $g(u,v,1)=[\t u: \t v:1]$ with
\begin{equation}\label{QRT root gen}
 \widetilde{u}=v, \quad \widetilde{v}=\frac{\alpha uv+\beta u-1}{u-\alpha v-\beta},
\end{equation}
which admit an integral of motion
\begin{equation}\label{integral QRT root gen}
K(u,v)=\frac{\alpha(\alpha+1)(u^2+v^2-1)-(\alpha+1)uv+\beta(u+v)-\beta^2}{(u^2-1)(v^2-1)}.
\end{equation}
Note that map \eqref{QRT root} corresponds to $\alpha=1/2$, $\beta=-1/2$. As a characteristic feature we choose the existence of a short singular orbit  (the third one in \eqref{123_QRT_orbits}):
$$
(q_8q_{11}) \longrightarrow q_{11} \longrightarrow  (q_7q_{11}),
$$
i.e., of a point $q_{11}$ which belongs both to $\mathcal I(g)$ and to $\mathcal I(g^{-1})$. One easily computes:
$$
\mathcal I(g)=\Big\{q_8, \Big(1,\frac{1-\beta}{\alpha}\Big), \Big(-1,-\frac{1+\beta}{\alpha}\Big)\Big\},
$$
$$
\mathcal I(g^{-1})=\Big\{q_7, \Big(\frac{1-\beta}{\alpha},1\Big), \Big(-\frac{1+\beta}{\alpha},-1\Big)\Big\},
$$
where
$$
q_7=[0:1:0], \quad q_8=[1:0:0].
$$
We have a one-parameter generalization of the previous case, with
$$
q_{11}=(-1,-1)\in \mathcal I(g)\cap \mathcal I(g^{-1}),
$$
under the condition
\begin{equation}\label{cond}
\alpha=1+\beta.
\end{equation}
In what follows, we parametrize the coefficients $\alpha$, $\beta$ according to
\begin{equation}\label{parameter}
\beta=\frac{c-1}{2}, \quad \alpha=\frac{c+1}{2}.
\end{equation}

\begin{proposition}
Under condition \eqref{parameter}, the map $g$ given in \eqref{QRT root} has three confined singular orbits as in \eqref{123_QRT_orbits}. Moreover, the point $q_3$ is infinitely near to $q_{11}$ (with the slope $-1$). The map $g$ has a fixed point
$$
q_9=\Big(\frac{1}{c},\frac{1}{c}\Big).
$$
The pencil of invariant curves $\{K(u,v)=\lambda\}$ of the map $g$ is as in \eqref{pencil quartic}. The eight finite base points $q_1,\ldots,q_6,q_{10},q_{11}$ lie on the conic given by the numerator of $K(u,v)$.
\end{proposition}
\begin{proof} The second singular orbit in \eqref{123_QRT_orbits} is confirmed by an easy computation. Let us compute the first (long) singular orbit, starting with the remaining point from $\mathcal I(g^{-1})$, that is, with
$$
q_{10}=\Big(\frac{3-c}{1+c},1\Big)\in \mathcal I(g^{-1}).
$$
We compute:
$$
g(q_{10})=q_1=\Big(1,-\frac{1-c}{3+c}\Big), \quad g(q_{1})=q_2=\Big(-\frac{1-c}{3+c},-1\Big),
$$
$$
g(q_2)=q_3>q_{11}=(-1,-1)\;\text{with\;\ slope}\; -1,
$$
$$
g(q_3)=q_4=\Big(-1,-\frac{1-c}{3+c}\Big), \quad g(q_{4})=q_5=\Big(-\frac{1-c}{3+c},1\Big),
$$
$$
g(q_5)=q_6=\Big(1,\frac{3-c}{1+c}\Big)\in \mathcal I(g).
$$
One easily computes also that $g^{-1}$ blows up the point $q_{10}$ to the line $(q_6q_8)$, while $g$ blows up the point $q_6$ to the line $(q_7q_{10})$.

The fixed point $q_9$ is given by a straightforward computation (note that for $c\neq 0$, the point $q_9$ does not lie on $(q_7q_8)$, the line at infinity).
\end{proof}
All this is illustrated on Fig. \ref{Fig biquadratics pert}.
\medskip

%%%%%%%%%%%%%%%%%%%%%%%%%%%%%%%%%%%%%%%%%%%%%%%%%%%%%%%%%%%%%%%%%%%%
\begin{figure}
\begin{center}
\includegraphics[width=0.6\textwidth]{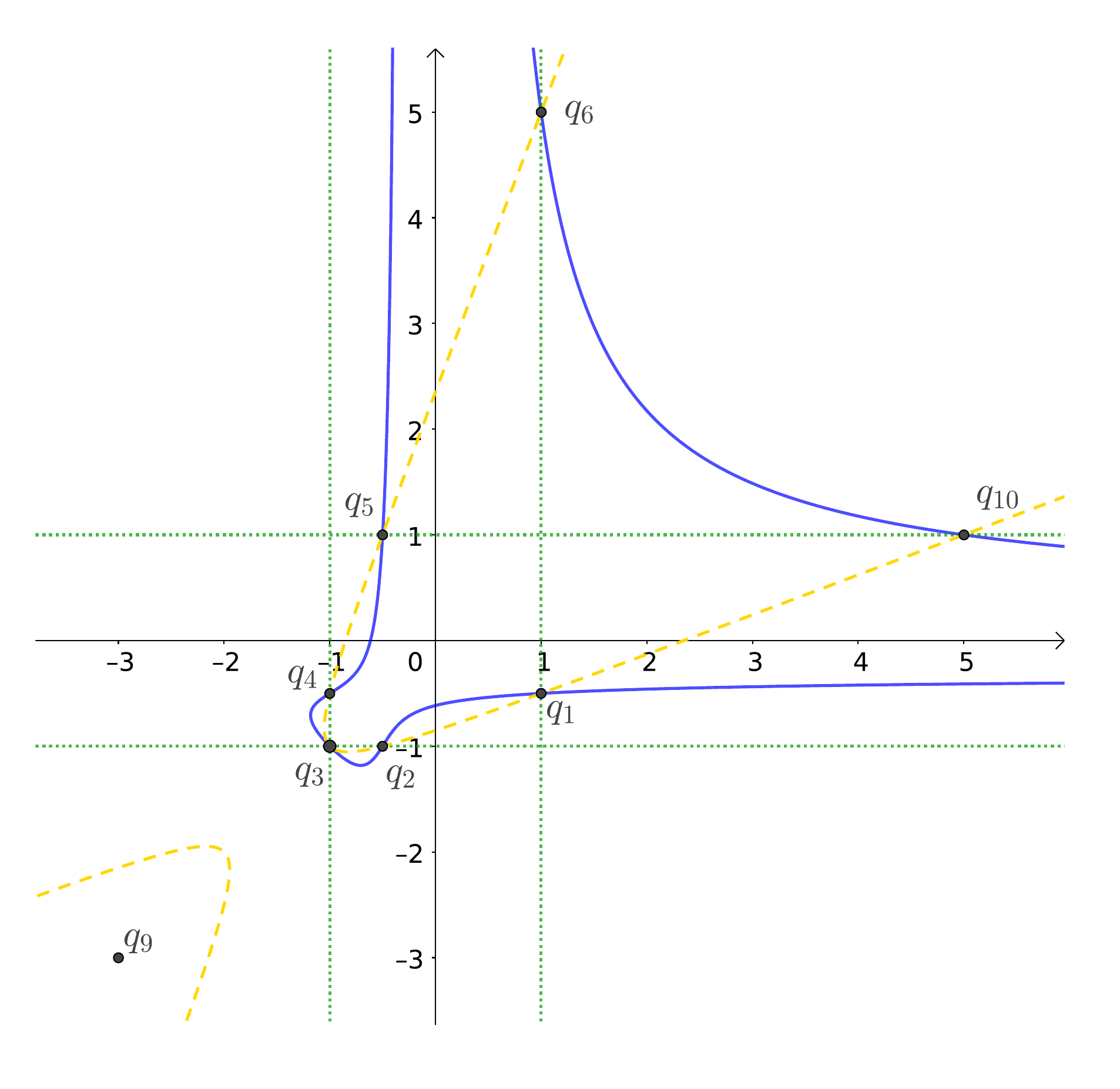}
\caption{Geometry of base points of the biquadratic pencil $\{K(u,v)=\lambda\}$ with $K(u,v)$ from \eqref{integral QRT root gen} with $\beta=-2/3$, $\alpha=1/3$, that is, $c=-1/3$. All biquadratics pass through $q_{11}=(-1,-1)$ with the slope $-1$ (corresponding to $q_3$). Yellow: the conic through the eight base points $q_1,\ldots,q_6,q_{10},q_{11}$. Green: four lines $\{u=\pm 1\}$, $\{v=\pm 1\}$.}
\label{Fig biquadratics pert}
\end{center}
\end{figure}
%%%%%%%%%%%%%%%%%%%%%%%%%%%%%%%%%%%%%%%%%%%%%%%%%%%%%%%%%%%%%%%%%%%%

There holds the following converse to Proposition \ref{th sextic to quartic}. We perform a quadratic Cremona change of variables based at $q_9,q_{10},q_{11}$ (recall that $q_9$ is not a base point of the pencil of invariant curves, while $q_{10}$ and $q_{11}$ are simple base points, the latter having an infinitely close base point $q_3$).
\begin{proposition} \label{th quartic to sextic}
Consider a quadratic Cremona map $\phi^{-1}$ blowing down the lines  $(q_{10}q_{11})$, $(q_9q_{11})$, $(q_9q_{10})$ to points denoted by $p_9,p_{10},p_{11}$, respectively, and blowing up the points $q_9,q_{10},q_{11}$ to the lines $(p_{10}p_{11})$, $(p_9p_{11})$, $(p_9p_{10})$. All other  base points $q_i$, $i=1,\ldots,8$ are regular points of $\phi^{-1}$ and their images are denoted by $p_i=\phi^{-1}(q_i)$.
The change of variables $\phi^{-1}$ maps the pencil \eqref{pencil quartic} of biquadratics to the pencil \eqref{pencil sextic} of sextic curves.
\end{proposition}
\begin{proof}
 The total image of a curve $C$ of the pencil \eqref{pencil quartic} is a curve of degree 8. Since $C$ passes through $q_{10},q_{11}$, its total image contains the lines $(p_{9}p_{11})$, $(p_9p_{10})$. Dividing by the linear defining polynomials of these two lines, we see that the proper image of $C$ is a curve of degree $6$. This curve passes through all points $p_i$, $i=1,\ldots,8$ (for $i=7,8$ with multiplicity 2). The curve $C$ of degree 4 intersects the line $(q_{10}q_{11})$ at two points $q_{10}$, $q_{11}$, and two further points, therefore its proper image passes through $p_9$ with multiplicity 2. On the other hand, the curve $C$ of degree 4 has three additional intersection points with each of the lines $(q_9q_{10})$ and $(q_9q_{11})$, different from the points $q_{10}$, respectively $q_{11}$. Therefore, its proper image passes through $p_{11}$, resp. $p_{10}$, with multiplicity 3.
\end{proof}

It remains to conjugate the QRT root $g$ by the quadratic change of variables $\phi^{-1}$.
\begin{proposition}
The map $f=\phi^{-1}\circ g\circ \phi$ is a quadratic Cremona map with three confined singular orbits, as in \eqref{123_orbits}. The eight base points $p_i$, $i=1,\ldots,6,10,11$ lie on a conic.
\end{proposition}
\begin{proof}
We have:
$$
p_i \overset{\phi}{\longrightarrow} q_i\overset{g}{\longrightarrow} q_{i+1}\xrightarrow{\phi^{-1}} p_{i+1}, \quad i=1,\ldots,5.
$$
Further,
$$
p_6 \overset{\phi}{\longrightarrow} q_6\overset{g}{\longrightarrow} (q_7q_{10}) \xrightarrow{\phi^{-1}} (p_7p_{10})
$$
(apply Lemma \ref{lemma Cremona} for $\phi^{-1}$);
$$
p_7 \overset{\phi}{\longrightarrow} q_7\overset{g}{\longrightarrow} q_{8}\xrightarrow{\phi^{-1}} p_{8};
$$
$$
p_8 \overset{\phi}{\longrightarrow} q_8\overset{g}{\longrightarrow} (q_{10}q_{11})\xrightarrow{\phi^{-1}} p_{9};
$$
$$
p_9 \overset{\phi}{\longrightarrow} (q_{10}q_{11}) \overset{g}{\longrightarrow} (q_1q_{10}) \xrightarrow{\phi^{-1}} (p_1p_{10})
$$
(apply Lemma \ref{lemma Cremona} first for $g$, then for $\phi^{-1}$);
$$
p_{10} \overset{\phi}{\longrightarrow} (q_9q_{11}) \overset{g}{\longrightarrow} (q_9q_{10}) \xrightarrow{\phi^{-1}} p_{11}
$$
(apply Lemma \ref{lemma Cremona} for $g$, taking into account that $q_9$ is a fixed point);
$$
p_{11} \overset{\phi}{\longrightarrow} (q_{9}q_{10}) \overset{g}{\longrightarrow} C(q_9,q_1,q_{10},q_7,q_{11}) \xrightarrow{\phi^{-1}} (p_1p_7)
$$
(apply Lemma \ref{lemma Cremona} first for $g$, taking into account that $q_9$ is a fixed point and $q_1=g(q_{10})$; then, the total $\phi^{-1}$-image of the conic is a curve of degree 4; however, three lines split off, being the blow-ups of $q_9$, $q_{10}$, $q_{11}$; thus, the proper image is the line through $\phi^{-1}(q_1)=p_1$ and $\phi^{-1}(q_7)=p_7$).

Next, we consider lines which are blown down by $f$:
$$
(p_6p_9) \overset{\phi}{\longrightarrow} (q_6q_9) \overset{g}{\longrightarrow} (q_{11}q_9) \xrightarrow{\phi^{-1}} p_{10}
$$
(apply Lemma \ref{lemma Cremona} first for $\phi$, then for $g$, taking into account that $q_9$ is a fixed point);
$$
(p_6p_{11}) \overset{\phi}{\longrightarrow} (q_6q_{11}) \overset{g}{\longrightarrow} q_7 \xrightarrow{\phi^{-1}} p_{7}
$$
(apply Lemma \ref{lemma Cremona} for $\phi$);
$$
(p_9p_{11}) \overset{\phi}{\longrightarrow} q_{10} \overset{g}{\longrightarrow} q_1 \xrightarrow{\phi^{-1}} p_1.
$$

It remains to show that the points $p_1,\ldots,p_6,p_{10},p_{11}$ lie on a conic. For this, we observe that the total $\phi^{-1}$-image of the conic $C$ through $q_1,\ldots,q_6,q_{10},q_{11}$ is a curve of $\deg=4$, from which two lines split off (blow-ups of $q_{10}$, $q_{11}$). Thus, the proper image is a conic. This conic contains $p_1=\phi^{-1}(q_1),\ldots,p_6=\phi^{-1}(q_6)$. It also contains $p_{10}$ and $p_{11}$ as the consequence of the fact that $C$ has additional intersection points with both blown-down lines $(q_9q_{11})$ and $(q_9q_{10})$, apart from $q_{11}$ and $q_{10}$, respectively.
\end{proof}

To make concrete computations, we normalize $\phi^{-1}$ by the following conditions:
\begin{equation}\label{p9 10 11}
 p_9=(0,1), \quad p_{10}=(-1,0), \quad p_{11}=(1,0),
\end{equation}
and
\begin{equation}\label{p7}
p_7 = (0,-1).
\end{equation}
Then a straightforward computation gives:
\begin{equation}
\begin{split}
\phi(x,y,z) = \Big[ & (1+2c-c^2) x y - x^2 - y^2 + z^2 :\\
& 2 c x y -(1-c^2) y z  -  c^2 y^2 - x^2 + z^2 :\\
& (1+c^2) x y - c(x^2+y^2-z^2) \Big],
\end{split}
\end{equation}
and for the further base points given by $p_i = \phi^{-1}(q_i)$, we find:
\begin{align}
p_1 &= -p_6 = \left(\frac{(1+c)(3+c)}{5-c^2}, \frac{4}{5-c^2}\right), \label{p1 6}\\
p_2 &= -p_5 = \left(\frac{(1+c) (3-c)}{ (1-c)(3+c)}, \frac{4}{(3+c)(1-c)}\right), \label{p2 5}\\
p_3 &= -p_4 = \left(\frac{3+c}{1-c}, \frac{4}{1-c}\right), \label{p3 4}\\
p_8 &= [c:1:0]. \label{p8}
\end{align}

\begin{theorem}
The map $f=\phi^{-1}\circ g\circ \phi$ is given by
\begin{equation}
\begin{split}
f= \Big[ & \big( x-c(y-z)\big) \big(z-x+(2+c) y \big):\\
&  y \big( (1+c) (z-y)+(3-c) x \big):\\
& z^2+(1-c) xz - (2-c) x^2+(1+c)(2-c) xy  - y^2  \Big].
\end{split}
\end{equation}
In the non-homogeneous coordinates, the map $f(x,y,1)=[\t x: \t y: 1]$ satisfies the following bilinear (Kahan-type) relations:
\begin{equation} \label{Kahan adj}
\left\{\begin{aligned}
\t x - x &= -(2-c)x \t x - c(2+c) y \t y + (1+c)(\t x y + x \t y)+ c,\\
\t y - y &= - (1+c)y \t y + (\t x y + x \t y).
\end{aligned}
\right.
\end{equation}
It possesses an integral of motion
\begin{equation}\label{H1 pert}
H_1(x,y)=\frac{\big(C_1(x,y)\big)^2 C_2(x,y)}{\big(1-(x-cy)^2\big)\big(1-(x-y)^2\big)\big(1-(x^2+y^2-2cxy)\big)},
\end{equation}
where
\begin{eqnarray}
C_1(x,y) & = & (1+c^2) x y + c (1-x^2-y^2), \\
C_2(x,y) & = & -2(1-c-c^2)xy+\frac{1}{2}(3-c-3c^2-c^3)y^2-cx^2+c.
\end{eqnarray}
The base points of the pencil of invariant curves $\{H_1(x,y)=\lambda\}$ are given in \eqref{p9 10 11}, \eqref{p7}, and \eqref{p1 6}--\eqref{p8}.
The conic $\{C_1(x,y)=0\}$ passes through $p_7,p_8,p_9,p_{10},p_{11}$ (it is the $\phi^{-1}$-image of the line $(q_7q_8)$), while the conic $\{C_2(x,y)=0\}$ passes through eight base points $p_1,\ldots,p_6,p_{10},p_{11}$.
\end{theorem}
\begin{proof}
A straightforward symbolic computation.
\end{proof}

On Fig. \ref{Fig pencil 123 perturbed} one can see several invariant curves $\{H_1(x,y)=\lambda\}$ of the map $f$.

%%%%%%%%%%%%%%%%%%%%%%%%%%%%%%%%%%%%%%%%%%%%%%%%%%%%%%%%%%%%%%%%%%%%
\begin{figure}
\begin{center}
\includegraphics[width=0.8\textwidth]{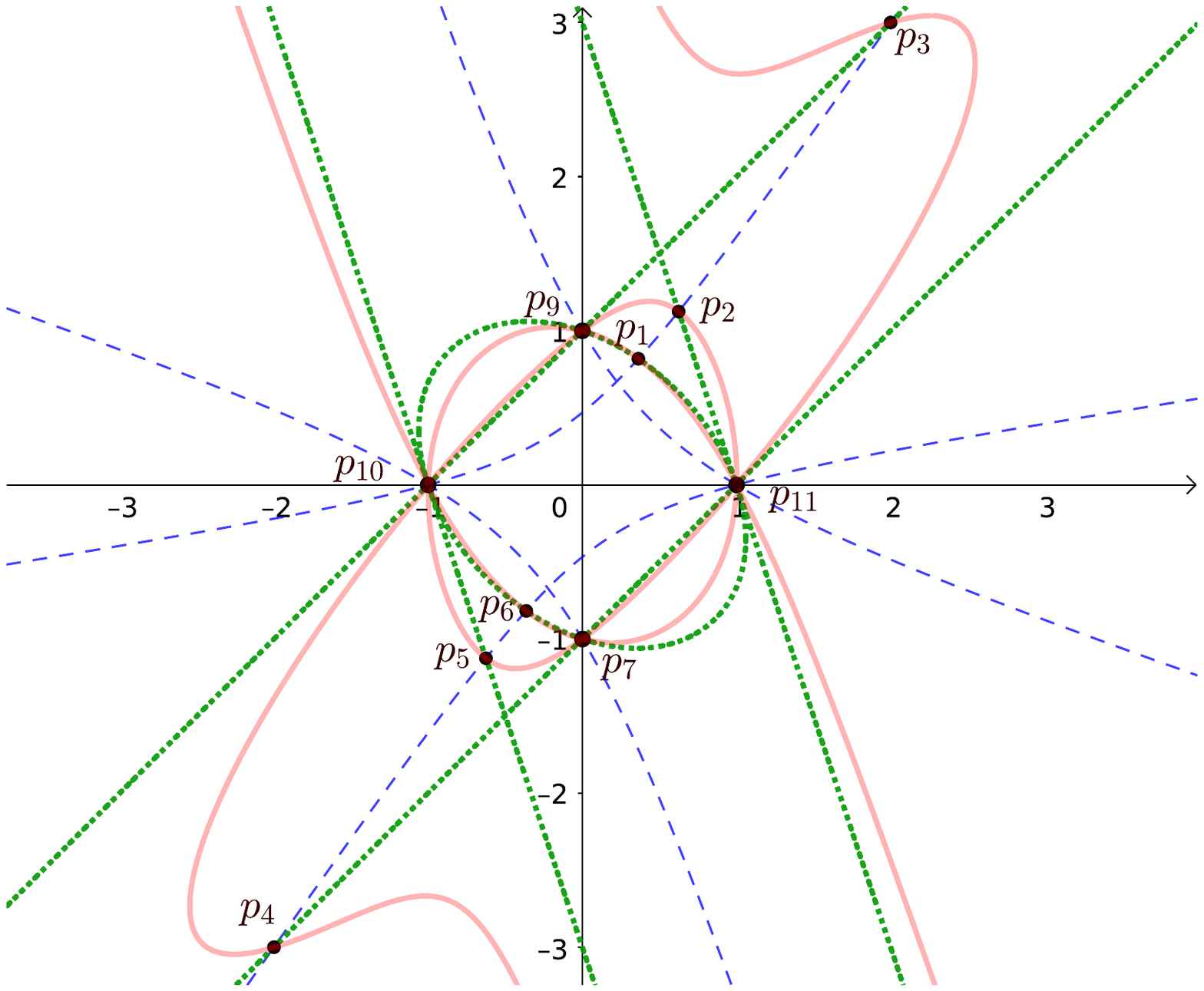}
\caption{Some invariant curves of the sextic pencil $H_1(x,y)=\lambda$ with $H_1(x,y)$ from \eqref{H1 pert}, $c=-1/3$. Blue: two conics (one of them with multiplicity 2) -- the vanishing set of the numerator of $H_1(x,y)$, i.e., $\lambda=0$. Green: four lines and a conic -- the vanishing set of the denominator of $H_1(x,y)$, i.e., $\lambda=\infty$. Red: a generic curve.}
\label{Fig pencil 123 perturbed}
\end{center}
\end{figure}
%%%%%%%%%%%%%%%%%%%%%%%%%%%%%%%%%%%%%%%%%%%%%%%%%%%%%%%%%%%%%%%%%%%%

%%%%%%%%%%%%%%%%%%%%%%%%%%%%%%%%%%%%%%%%%%%%%%%%
\section{Continuous limit}
\label{sect cont limit}

Re-scaling $(x,y)\mapsto(\epsilon x,\epsilon y)$ and $c\mapsto \epsilon^2 c$, we arrive at system \eqref{Kahan adj intro}, which in the limit $\epsilon\to 0$ is a discretization of system \eqref{123 pert intro}. The latter can be written as
\begin{equation}
\begin{pmatrix} \dot x \\ \dot y \end{pmatrix}
=
\frac{1}{(xy+c)y} \begin{pmatrix} \partial H/\partial y \\ -\partial H/\partial x\end{pmatrix}
\end{equation}
with
\begin{equation}\label{H pert}
H(x,y) = (xy+c)^2\left(-\frac{2}{3}xy+\frac{1}{2}y^2+\frac{1}{3}c\right).
\end{equation}
This is a one-parameter (inhomogeneous) perturbation of system \eqref{123 intro}. Like for the unperturbed system, all level sets $\{H(x,y)=\lambda\}$ of the integral of motion \eqref{H pert} are curves of genus 1 (and of degree 6). Thus, map \eqref{Kahan adj intro} is a non-trivial integrable Kahan-type discretization of \eqref{123 pert intro}.

Integrability of map \eqref{Kahan adj intro} is in a contrast to non-integrability of the straightforward Kahan discretization \eqref{123 pert Kahan} of \eqref{123 pert intro}.
\begin{proposition}
The map $f$ generated by bilinear equations \eqref{123 pert Kahan} is non-integrable, in the sense that its singularities are not confined.
\end{proposition}
\begin{proof} To show this, we restrict ourselves to the case $\epsilon=1$. The resulting quadratic Cremona map has three singularities, $p^+=(1,0)$ and two further points not lying on the line $\{y=0\}$. Likewise, the inverse map has three singularities, $p^-=(-1,0)$ and two further points not lying on the line $\{y=0\}$. Observe that the line $\{y=0\}$ is invariant. Thus, for the singularities to be confined, we need that some $f^n(p^-)=p^+$ for some $n\in\mathbb N$. The restriction of the map to the line $\{y=0\}$ is given by $\t x -x =-2x\t x+c$, or $\t x=\varphi(x)=(x+c)/(2x+1)$. One easily sees that, for a generic $c$, the orbit of $x=-1$ under this M\"obius transformation does not hit $x=1$. Indeed, $\varphi^n(-1)=1$ is a polynomial equation of degree $n$ for $c$. Thus, for all $c$ but a countable set this equation is not satisfied for any $n\in\mathbb N$.
\end{proof}
%%%%%%%%%%%%%%%%%%%%%%%%%%%%%%%%%%%%%%%%%%%%%%%%%%%

\section{Conclusions}

The results of the present paper confirm that the phenomenon discovered and described in \cite{PSZ} is not isolated, namely that in case of non-integrability of the standard Kahan discretization (when applied to an integrable system), its coefficients can be adjusted to restore integrability. Recall that the definition of Kahan's discretization includes a very straightforward dependence on the small stepsize $\epsilon$. Namely, it only appears in the denominator of the differences $(\t x-x)/\epsilon$ which approximate the derivatives $\dot x$, compare \eqref{eq: diff eq gen} and \eqref{eq: Kahan gen}. On the contrary, coefficients of the bilinear expressions on the right hand side of \eqref{eq: Kahan gen} are traditionally taken to literally coincide with the coefficients of the quadratic vector fields on the right hand side of \eqref{eq: diff eq gen}. This discretization method preserves integrability much more frequently than one would expect a priori, but not always. Our examples show that, if the straightforward recipe fails to preserve integrability, certain adjustments of the coefficients by quantities of the magnitude $O(\epsilon^2)$ may allow to restore integrability. Further extending the list of examples and finding their systematic explanation in terms of addition laws on Abelian varieties remains an important and entertaining task for the future.

This work was done in the frame of a summer research project of MS and YT at Technische Universit\"at Berlin in the Summer-Fall 2020 (which, due to the pandemic, was performed online). Research of YS is supported by the DFG Collaborative Research Center TRR 109 ``Discretization in Geometry and Dynamics''.

%%%%%%%%%%%%%%%%%%%%%%%%%%%%%%%%%%%%%%%%%%%%%%%%%%
{}

\begin{thebibliography}{}

\bibitem{BV}
M.P. Bellon, C.-M. Viallet,
{\em Algebraic entropy},
Commun. Math. Phys. {\bf 204} (1999), No. 2, 425--437.

\bibitem{CMOQ1}
E. Celledoni, R.I. McLachlan, B. Owren, G.R.W. Quispel.
{\em Geometric properties of Kahan's method},
J. Phys. A {\bf 46} (2013), 025201, 12 pp.

\bibitem{CMOQ2}
E. Celledoni, R.I. McLachlan, D.I. McLaren, B. Owren, G.R.W. Quispel.
{\em Integrability properties of Kahan's method},
J. Phys. A {\bf 47} (2014), 365202, 20 pp.

%\bibitem{CMOQ3}
%E. Celledoni, R.I. McLachlan, D.I. McLaren, B. Owren, G.R.W. Quispel.
%{\em Discretization of polynomial vector fields by polarization},
%Proc. R. Soc. A {\bf 471} (2015), 20150390, 10 pp.

\bibitem{CMOQ4}
E. Celledoni, R.I. McLachlan, D.I. McLaren, B. Owren, G.R.W. Quispel.
{\em Two classes of quadratic vector fields for which the Kahan map is integrable},
MI Lecture Note, Kyushu University {\bf 74} (2016), 60--62.

\bibitem{CMOQ5}
E. Celledoni, D.I. McLaren, B. Owren, G.R.W. Quispel.
{\em Geometric and integrability properties of Kahan's method: the preservation of certain quadratic integrals},
J. Phys. A {\bf 52} (2019), 065201, 9 pp.

\bibitem{D1}
J. Diller.
{\em Dynamics of birational maps of $\mathbb P^2$},
Indiana Univ. Math. J. {\bf 45} (1996), No. 3, 721--772.

\bibitem{D2}
J. Diller.
{\em Cremona transformations, surface automorphisms, and plane cubics},
Michigan Math. J. {\bf 60} (2011), 409--440.

\bibitem{DF}
J. Diller, C. Favre.
{\em Dynamics of bimeromorphic maps of surfaces},
Am. J. Math. {\bf 123} (2001), No. 6, 1135--1169.

\bibitem{Dui}
J.J. Duistermaat. {\em Discrete Integrable Systems. QRT Maps and Elliptic Surfaces},
Springer, 2010, xii+627 pp.


\bibitem{HK}
R.~Hirota, K.~Kimura.
{\em Discretization of the Euler top},
J. Phys. Soc. Japan {\bf 69} (2000), No. 3, 627--630.

\bibitem{HMM}
N.J.\ Hitchin , N.S.\ Manton, M.K.\ Murray,
{\em Symmetric monopoles},
Nonlinearity {\bf 8} (1995), No. 5, 661--692.


\bibitem{KCMMOQ}
P.H. van der Kamp, E. Celledoni, R.I. McLachlan, D.I. McLaren, B. Owren, G.R.W. Quispel.
{\em Three classes of quadratic vector fields for which the Kahan discretization is the root of a generalised Manin transformation},
J. Phys. A: Math. Theor. {\bf 52} (2019) 045204.

\bibitem{KMQ}
P.H. van der Kamp, D.I. McLaren, G.R.W. Quispel.
{\em  Generalised Manin transformations and QRT maps},
J. Comput. Dyn. {\bf 8} (2021), No. 2, 183--211.

\bibitem{K}
W.~Kahan.
{\em Unconventional numerical methods for trajectory calculations},
Unpublished lecture notes, 1993.

\bibitem{KH}
K.~Kimura, R.~Hirota.
{\em Discretization of the Lagrange top},
J. Phys. Soc. Japan {\bf 69} (2000), No. 10, 3193--3199.

\bibitem{PPS1}
M.~Petrera, A.~Pfadler, Yu.B.~Suris.
{\em On integrability of Hirota-Kimura-type discretizations: experimental study of the discrete Clebsch system},
Experimental Math. {\bf 18} (2009), No. 2, 223--247.

\bibitem{PPS2}
M.~Petrera, A.~Pfadler, Yu.B.~Suris.
{\em On integrability of Hirota-Kimura type discretizations},
Regular Chaotic Dyn. {\bf 16} (2011), No. 3-4, 245--289.

\bibitem{PPS3}
M.~Petrera, A.~Pfadler, Yu.B.~Suris (with appendix by Yu.N.~Fedorov).
{\em On the construction of elliptic solutions of integrable birational maps},
Experimental Math. {\bf 26} (2017), No. 3, 324--341.

\bibitem{PS}
M.~Petrera, Yu.B.~Suris.
{\em On the Hamiltonian structure of Hirota-Kimura discretization of the Euler top},
 Math. Nachr. {\bf 283} (2010), No. 11, 1654--1663.

\bibitem{PS2}
M.~Petrera, Yu.B.~Suris.
{\em A construction of a large family of commuting pairs of integrable symplectic birational 4-dimensional maps},
 Proc. Royal Soc. A, {\bf 473} (2017),  20160535, 16 pp.

\bibitem{PS3}
M.~Petrera, Yu.B.~Suris.
{\em New results on integrability of the Kahan-Hirota-Kimura discretizations}. - In: {\em Nonlinear Systems and Their Remarkable Mathematical Structures}, Ed. N. Euler, CRC Press, Boca Raton FL, 2018, p. 94--120.

\bibitem{PS4}
M.~Petrera, Yu.B.~Suris.
{\em Geometry of the Kahan discretizations of planar quadratic Hamiltonian systems. II. Systems with a linear Poisson tensor},
J. Comput. Dyn., {\bf 6} (2019), 401--408.

\bibitem{PSS}
M.~Petrera, J.~Smirin, Yu.B.~Suris.
{\em Geometry of the Kahan discretizations of planar quadratic Hamiltonian systems},
Proc. Royal Soc. A, {\bf 475} (2019), 20180761, 13 pp.

\bibitem{PSWZ}
M. Petrera, Yu.B. Suris, Kangning Wei, R. Zander.
{\em Manin involutions for elliptic pencils and discrete integrable systems},
Math. Phys. Anal. Geom., {\bf 24} (2021), No. 6, 26 pp.


\bibitem{PZ}
M.~Petrera, R.~Zander,
{\em New classes of quadratic vector fields admitting integral-preserving Kahan-Hirota-Kimura discretizations},
J. Phys. A: Math. Theor. {\bf 50} (2017) 205203, 13 pp.

\bibitem{PSZ}
M. Petrera, Yu.B. Suris, R. Zander,
{\em How one can repair non-integrable Kahan discretizations},
 J. Phys. A: Math. Theor. {\bf 53} (2020) 37LT01, 7 pp.

\bibitem{QRT}
G.R.W. Quispel, J.A.G. Roberts, C.J. Thompson.
{\em  Integrable mappings and soliton equations II},
Physica D {\bf 34} (1989) 183--192.

 \bibitem{SS}
J.M.~Sanz-Serna.
{\em An unconventional symplectic integrator of W. Kahan},
Appl. Numer. Math. {\bf 16} (1994), 245--250.

\bibitem{S}
Yu.B. Suris. \emph{The Problem of Integrable Discretization: Hamiltonian Approach}. Progress in Mathematics, Vol. 219.
Basel: Birkh\"auser, 2003. xxi+1070 pp.

%\bibitem{Ves}
%A.P. Veselov.
%{\em Integrable maps},
%Russ. Math. Surv. {\bf 46} (1991) 3--45.

\bibitem{Z}
R. Zander.
{\em On the singularity structure of Kahan discretizations of a class of quadratic vector fields},
European J. Math. (to appear), {\tt     arXiv:2003.01659 [nlin.SI]}.


\end{thebibliography}
\end{document}